\documentclass[letterpaper]{article} 
\usepackage{aaai25}  
\usepackage{times}  
\usepackage{helvet}  
\usepackage{courier}  
\usepackage[hyphens]{url}  
\usepackage{graphicx} 
\usepackage{natbib}  
\usepackage{caption} 
\usepackage{algorithm}
\usepackage{algorithmic}
\usepackage{amsmath,amssymb,amsfonts,amsthm,epsfig,nccmath}
\usepackage{amsfonts}
\usepackage{dsfont}
\usepackage{textcomp,subcaption}
\usepackage{ifthen}
\usepackage{color}
\usepackage{appendix}
\newtheorem{definition}{Definition}
\newtheorem{theorem}{Theorem}
\newtheorem{lemma}{Lemma}
\newtheorem{assumption}{Assumption}
\newtheorem{remark}{Remark}
\newtheorem{claim}{Claim}

\newcommand{\prob}[1]{\mathsf{Pr}\!\left(#1\right)}
\newcommand{\EXP}[1]{\mathsf{E}\!\left[#1\right]}

\newcommand{\remove}[1]{}
\newcommand{\Ind}[1]{\mathds{1}_{\left\{ {#1} \right\} }}

\newcommand{\ignore}[1]{}

\newcommand{\R}{\mathbb{R}}

\newcommand{\nn}{\nonumber}
\newcommand{\norm}[1]{\| #1 \|}

\newcommand{\st}{$P_1$}
\newcommand{\ins}{$P_A$}
\newcommand{\calcZ}{$Calc_Z$}
\newboolean{showcomments}
\setboolean{showcomments}{true}
\newcommand{\jk}[1]{ \ifthenelse{\boolean{showcomments}}
	{\textcolor{red}{(JK says: #1)}} {} }
\newcommand{\hs}[1]{\ifthenelse{\boolean{showcomments}} 
	{\textcolor{magenta}{(HS says: #1)}} {} }
\newcommand{\dm}[1]{\ifthenelse{\boolean{showcomments}} 
	{\textcolor{blue}{(DM says: #1)}} {} }

\nocopyright
\title{Blotto on the Ballot: A Ballot Stuffing Blotto Game}
\author {
	Harsh Shah\textsuperscript{\rm 1},
	Jayakrishnan Nair\textsuperscript{\rm 1},
	D Manjunath\textsuperscript{\rm 1},
	Narayan Mandayam\textsuperscript{\rm 2}
}
\affiliations {
	\textsuperscript{\rm 1}Department of Electrical Engineering, IIT Bombay, Mumbai, India\\
	\textsuperscript{\rm 2}WINLAB, Rutgers University, North Brunswick, New Jersey 08902, USA\\
	\{210100063, jayakrishnan.nair, dmanju\}@iitb.ac.in, narayan@winlab.rutgers.edu
}

\usepackage{bibentry}

\begin{document}
	
	\maketitle
	
	\begin{abstract}
		We consider the following Colonel Blotto game between parties $P_1$ and $P_A.$ $P_1$ deploys a non negative number of troops across $J$ battlefields, while $P_A$ chooses $K,$ $K < J,$ battlefields to remove all of $P_1$'s troops from the chosen battlefields. $P_1$ has the objective of maximizing the number of surviving troops while $P_A$ wants to minimize it. Drawing an analogy with ballot stuffing by a party contesting an election and the countermeasures by the Election Commission  to negate that, we call this the \emph{Ballot Stuffing Game.} For this zero-sum resource allocation game, we obtain the set of Nash equilibria as a solution to a convex combinatorial optimization problem. We analyze this optimization problem and obtain insights into the several non trivial features of the equilibrium behavior. These features in turn allows to describe the structure of the solutions and efficient algorithms to obtain then. The model is described as ballot stuffing game in a plebiscite but has applications in security and auditing games. The results are extended to a parliamentary election model. Numerical examples illustrate applications of the game. 
	\end{abstract}
	
	%
	\section{Introduction}
	\label{sec:intro}
	Resource allocation games over multiple asymmetric battlefields between  asymmetric competing players has wide applications in electoral politics, e.g., \cite{Washburn13}, security deployments, e.g., \cite{Boehmer2024}, market strategies, e.g., \cite{Maljkovic2024}), etc. The asymmetry could be in the deployment cost in and/or the payoffs from the different battlefields and also in the strategy space for the two players. A popular model is the Colonel Blotto Game. The classical version is a two-player, zero-sum, multi-dimensional resource allocation game over $J$ battlefields. The objective of the players is to maximize the number of battlefields that they win.  Players have finite resources that are distributed among the $J$ battlefields without knowing the opponents deployment. Each battlefield is won by the player that deploys the higher number of resources for that location.  
	We describe and analyze a new variation that we call the \textit{Ballot Stuffing Game.}
	
	Our illustrative scenario is a game between an election contestant (attacker) and the Election Commission (the defender) that is responsible for a fair election. Party $P_1$ is contesting the election and attempts to gain an unfair advantage over its rivals by indulging in `ballot stuffing,' i.e., by having votes cast in its favour by non eligible voters, voters voting multiple times, or through other means. The Election Commission, denoted by $P_A,$ is aware of this intention and aims to mitigate the negative effect by placing observers in a subset of the booths to prevent stuffing. 
	The resources of the attacker $P_1$ above is a positive real number and it incurs a deployment dependent cost. The resources of the defender called $P_A$ above, is an integer and the possible deployment in each battlefield is a binary integer. The effect of $P_A$ deploying a resource in a battlefield is to nullify all the resources of $P_1$ in that battlefield. In the next section, we formally describe the model and develop the notation. We will also introduce two variations of the \emph{Ballot Stuffing Game.}
	
	Ballot stuffing, while not widespread, is a known threat in all forms of elections in all countries. There have been accusations of election fraud in mayoral elections in the US\cite{Bridgeport}, the Russian elections \cite{Harvey2016,Russia-Video}, and electronic voting over the Internet in organisational elections \cite{Cortier2011}. It is also not uncommon in reputation systems \cite{Bhattacharjee05}. However, it has been observed that the presence of observers or inspectors reduces the levels of fraud \cite{Asunka19}.  
	
	Colonel Blotto games have been studied since, at least 1912 when it was first posed by Emile Borel. A key insight was developed in \cite{Roberson06} where the unique equilibrium payoffs and the distributions of the deployments were characterized in some detail. Since then several contributions to the algorithms to compute the equilibrium have been developed in \cite{Ahemdinejad16}, \cite{Behnezhad17}, \cite{Behnezhad2018}, \cite{Behnezhad18b}, \cite{Behnezhad19}, and \cite{Behnezhad23} among others. There are also some similarities of the Ballot Stuffing Game to the Auditing Games games \cite{Behnezhad19} and Blotto Game with Testing \cite{Sonin2024}. However, the resources available to the attacker and the defender and also the rules of engagement in the battlefields in these games are significantly different. 
	
	\section{Model and Notation}
	\label{sec:Model}
	There are $J$ polling stations and two players ($P_1,P_2$) contesting for the votes that will be cast at these stations. The actual number of votes that will be cast for Player $P_i$ in station $j$ is a random variable $X_{i,j},$ with $i \in \{1,2\}$ and $1 \leq j \leq J.$ Define $X_j := X_{1,j}-X_{2,j}$ as the excess votes for $P_1$ that will be cast in station $j.$ 
	
	Let $V$ be the event that $P_1$ is the victor in the election. There are two election models that we will consider. In the Plebiscite model
	\begin{equation}
		V = \Ind{\sum_{j=1}^J X_j > 0}
		\label{eq:pleb-V}
	\end{equation}
	i.e., $P_1$ wins the election if the total of the votes cast for it from all the stations exceeds that of $P_2.$ We could also define a generalized Parliamentary model with $w_j$ being the weight associated with polling station $j$ and 
	\begin{equation}
		V = \Ind {\left( \sum_{j=1}^J w_j \Ind { X_j > 0 }\right) > W/2 }.
		\label{eq:parl-V2}
	\end{equation}
	Here, $P_1$ is the winner if the weighted sum of the polling stations won by $P_1$ is more than half of the sum of the weights. In the UK and many countries $w_j=1$ while the US has different weights for the different stations. 
	
	
	$P_1$ can try to increase its vote count in station $j$ (without any corresponding response from $P_2$) by, say, $z_j$ at a cost of $g_j(z_j).$ We assume the cost function $g_j(z_j)$ is strictly convex, differentiable, and increasing in its argument, beginning from zero when its argument is zero, for $1 \leq j \leq J.$ Define $z:=[z_1, \ldots, z_J]$ to be the vector of votes stuffed by $P_1$ in each of the stations, $g(z):=[g_1(z_1), \ldots, g_J(z_J)]$ to be the vector of cost functions and $V(z):= \Ind{\sum_{j=1}^J (X_j+z_j) > 0}.$ Let $\phi(z)$ be the probability that $P_1$ wins the election when it stuffs according to $z.$ For the plebiscite model $\phi(z)$ will be given by
	\begin{align}
		\phi(z) := \EXP{V(z)} & = \prob{ \sum_{j=1}^J (X_j + z_j) > 0 },
		\label{eq:pleb-phi}
	\end{align}
	and for the parliamentary model, we will have
	\begin{align}
		\phi(z) := \EXP{V(z)} & = \prob{ \ \left( \sum_{j=1}^J
			\Ind{ (X_j+z_j) > 0 }\right) > J/2 }.
		\label{eq:parl-phi}
	\end{align}
	$\phi(0)$ is the prior probability of $P_1$ winning.
	
	$P_1$ has a total budget of $G$ and, assuming the costs are additive, has the following objective.	%
	\begin{align}
		& \max_{z} \ &\quad & \ \phi(z) \nonumber \\
		& \mbox{subject to } && \sum_{j=1}^J g_j(z_j) \ \leq \ G, \nonumber \\    
		&  && 0 \leq z_j \hspace{0.3in} \mbox{for $1 \leq j \leq J$},
		\label{eq:P1-Obj}
	\end{align}
	
	\subsection{The Ballot Stuffing Game}
	Recall that we had assumed that $P_2$ does not respond to the stuffing but the Election Commission, the polling oversight authority, denoted by $P_A,$ is aware of the stuffing intentions of $P_1$ and has $K$, $K <
	J$ inspectors that can be deployed at the polling stations. If an
	inspector is deployed at polling station $j,$ then $P_1$ does not
	stuff at that station and the $z_j$ that was planned for station $j$
	is wasted. Let $y_j$ be the indicator variable that an inspector gets
	deployed at station $j.$ Thus the vote differential at station $j$ is
	$X_j + z_j (1-y_j).$
	
	Let $y:=[y_1, \ldots, y_J]$ be the deployment of inspectors by $P_A.$
	Let $\mathcal{Y}$ be the set of all feasible deployments of the inspectors, i.e.,
	\begin{displaymath}
		\mathcal{Y} = \left\{y: \sum_j y_j = K,\ y_j \in \{0,1\}\right\}. 
	\end{displaymath}
	With the deployment of inspectors according to $y,$ Define $V(z,y):=\Ind{\sum_{j=1}^J (X_j +z_j(1-y_j))> 0}.$ For
	the plebiscite model $\phi(z,y)$ will given by
	\begin{align}
		\phi(z,y) :&= \EXP{ V(z,y)} \nonumber \\ 
		& = \prob{ \sum_{j=1}^J (X_j + (1-y_j)z_j) > 0 },
		\label{eq:phi-zy-pleb}
	\end{align}
	and for the parliamentary model, it will be 
	\begin{align}
		\phi(z,y) :&= \EXP{V(z,y)} \nonumber \\
		& = \prob{ \sum_{j=1}^J \Ind { (X_j+(1-y_j)z_j) > 0 }  > J/2 }.
		\label{eq:phi-zy-parl}  
	\end{align}
	
	Hence this develops a dynamic game between $P_1$ and $P_A$ which is defined as
	\begin{definition} \label{def:ballot-game} 
		The Ballot Stuffing Game $\mathcal{B}(g,G,K)$ is the full information game between $P_1$ and $P_A$ where $\{(\mu_j,\sigma_j\}_{j=1,\ldots J},$ $g,$ $G,$ and $K$ is known to both $P_1$ and $P_A.$ $P_1$ chooses $z$ according to 
		\begin{align}
			& \max_{z} \ & \quad & \ \phi(z,y) \nonumber \\
			& \mbox{subject to}  && \sum_{j=1}^J g_i(z_i) \ \leq \ G, \nonumber \\
			& && 0 \leq z_j \ \mbox{for $1 \leq j \leq J$},
			\label{eq:Game-P1-Obj}
		\end{align}
		and $P_A$ chooses $y$ according to
		\begin{align}
			& \min_{y} \ & \quad &  \phi(z,y) \nonumber\\   
			& \mbox{subject to:} && \sum_{j=1}^J y_i = K, \ \ \ \ y_j \in \{0,1\}.
			\label{eq:Game-PA-Obj}
		\end{align}
		
	\end{definition}
	
	\section{Plebiscite Model}
	\label{sec:plebiscite}
	
	In this section, we characterize the Nash equilibrium of the Ballot Stuffing Game between \st\ and \ins\ in the Plebiscite model.
	
	In the Plebiscite model, since $\phi(z,y)$ is  monotonically increasing in $\left( \sum_{i=1}^Jz_i \right)$, votes at all stations contribute equally to the objective function in \eqref{eq:Game-P1-Obj} and \eqref{eq:Game-PA-Obj} which then reduces to $\sum_{j=1}^J z_j(1-y_j).$ We now analyse equilibrium strategies for $P_1$ and $P_A$ in this ballot stuffing game.
	
	
	\subsection{A Stackelberg variation}
	\label{sec:Stackelberg}
	As a first step towards characterizing the Nash equilibrium of the Ballot Stuffing Game, we first analyze a Stackelberg variation of this game, with \st\ being the leader, and \ins\ the follower. We subsequently relate the optimization that the leader must perform in this Stackelberg game to the Nash equilibrium of the (simultaneous play) Ballot Stuffing Game.
	
	For any given action $z$ of \st\ (leader), the response by \ins\ (the follower) is trivial---place inspectors at those~$K$ booths where the most vote stuffing has occurred. Defining~$z_{(i)}$ as the $i$-th largest value of the stuffing vector~$z$ (i.e., $z_{(1)} \geq z_{(2)} \geq \cdots \geq z_{(J)},$ the equilibrium strategy on part of \st\ therefore solves the following optimization.
	\begin{align}
		\max_{z} \quad & \sum_{j= K+1}^J z_{(j)}
		\nonumber\\  
		\label{eq:Stackelberg1}
		\text{subject to} \quad & \sum_{j=1}^J g_j(z_j) \leq G \tag*{(Leader-Opt)} \\
		& 0 \leq z_j \ \mbox{ for } 1 \leq j \leq J \nn 
	\end{align}
	The following observations are now immediate.
	\begin{lemma}
		\label{lemma:Stackelber1}
		The optimization~\emph{\ref{eq:Stackelberg1}} admits a unique optimal solution~$z^*.$ 
		Moreover, 
		$$|\{i \in [J] \ \colon\ z^*_i = \norm{z^*}_{\infty}\} | \geq K+1,$$
		$$\sum_{j=1}^{J} g_j(z^*_j) = G,$$ i.e., the set of indices of~$z^*$ that have the maximum value has cardinality at least~$K+1$ and cost of constraints hold with equality.
	\end{lemma}
	The proof of the lemma is provided in the Appendix. To characterize other structural properties of the solution of~\ref{eq:Stackelberg1} using tools from convex optimization, it is convenient to express it in a standard form as follows.
	\begin{align}
		\min \quad & -U \nn \\
		\text{subject to} \quad & 0 \leq z_j, \quad
		j=1,\ldots,J \nn \\ 
		& U \leq z \ \cdot \ (\boldsymbol{1}-y) \ \ \ \forall\ y \in \mathcal{Y} \label{eq:Stackelberg2}  \tag*{(Leader-Opt-Std)} \\
		& \sum_{j=1}^{J} g_j(z_j) \leq G \nn
	\end{align}
	Here, $\boldsymbol{1}$ denotes a vector of ones, and the auxiliary variable~$U$ captures the number of votes that are `successfully stuffed.' That~\ref{eq:Stackelberg2} is a convex optimization is immediate. In the special case~$K = 0$ (i.e., there is no~\ins\ at all),~\ref{eq:Stackelberg2} can be solved via a standard water-filling approach~\cite{Boyd}. However, the case $K>0$ is considerably more involved,\footnote{It is also worth noting that since $|\mathcal{Y}| = {J \choose K},$ the optimization~\ref{eq:Stackelberg2} has combinatorially many constraints; we describe a novel and efficient waterfilling-inspired algorithm to solve this optimization in Section~\ref{sec:algo}.} as we now demonstrate below by identifying the structural properties that uniquely characterize the solution of~\ref{eq:Stackelberg2}.
	
	For a stuffing vector~$z \in \R_+^{J},$ define
	\begin{itemize} 
		\item $A_z = \{i \in [J] \colon z_i = \norm{z}_{\infty}\}$ as the set of indices that take the maximum value,
		\item $B_z = \{i \in [J] \colon 0 < z_i < \norm{z}_{\infty}\}$ as the set of indices that take a positive value strictly less than the maximum value,
		\item $C_z = \{i \in [J] \colon z_i = 0$ as the set of indices of $z$ that take the value zero.
	\end{itemize}
	Note that for a \emph{non-zero} vector~$z,$ the sets $A_z,$ $B_z$ and $C_z$ define a partition of~$[J].$
	
	\begin{theorem}[\textbf{Structure of $z^*$}]
		\label{thm:Stackelberg-structure}
		For non-zero vector $z \in \R_+^{J}$ satisfying~$|A_z| > K,$ define $$\theta_z := \frac{\sum_{j \in A_z}g_j'(z_j)}{|A_z|-K}.$$ Such a $z$ is the unique optimal solution of~\emph{\ref{eq:Stackelberg2}} if and only if the following conditions are satisfied.
		\begin{enumerate}
			
			\item \textbf{Slope Order Condition}: The slopes \{\( g'_i(z_i) \)\} follow the order:
			\[
			g'_j(z_j) \leq \theta_z \leq g'_{\ell}(z_{\ell}) \quad \forall j \in A_z, \ell \in C_z
			\]
			
			\item \textbf{Structure Slope Condition}: For every \( k \in B_z \):
			\[
			g'_k(z_k) = \frac{\sum_{j \in A_z} g'_j(z_j)}{|A_z| - K}
			\]
			
			\item \textbf{Sum Condition}: The cost constraint holds with equality:
			\[
			\sum_{i=1}^{J} g_i(z_i) = G
			\]
		\end{enumerate}
	\end{theorem}
	
	Theorem~\ref{thm:Stackelberg-structure} provides necessary and sufficient conditions for a stuffing vector to optimize~\ref{eq:Stackelberg2}, or equivalently, for a stuffing vector to constitute the equilibrium action of the leader in the Stackelberg game under consideration. The quantity~$\theta_z$ plays a key role in these conditions; we refer to it as the \emph{structure slope.} It may be interpreted as the marginal cost associated with successful vote stuffing from the collection of booths in~$A_z.$ Indeed, since \ins\ (the follower) would wipe out the vote stuffing planned in $K$ of the booths in~$A_z,$ a small~$\epsilon$ increment in the planned stuffing at each booth in~$A_z,$ at an incremental cost of~$\epsilon \sum_{i\in A_z}{g_i'(z_i)},$ would result in a $(|A_z| - K)\epsilon$ increment in successful vote stuffing. The \emph{structure slope condition} in Theorem~\ref{thm:Stackelberg-structure} states that under the optimal stuffing vector, the marginal cost of stuffing at any booth in~$B_z$ must match the structure slope (if $B_z$ is non-empty). Additionally, the \emph{slope order condition} states that the marginal cost of planned stuffing of any booth in~$A_z$ must be less than the structure slope, whereas the marginal cost of stuffing at the `unstuffed' booths in~$C_z$ must be greater. The proof of Theorem~\ref{thm:Stackelberg-structure} is based on invoking the KKT conditions for the optimization~\ref{eq:Stackelberg2} and a novel application of Farkas' lemma \cite{Boyd}; the details are in the Appendix.
	
	\subsection{Nash equilibrium under the Plebiscite model}
	
	Having characterized the equilibrium of a leader-follower variant of the Ballot Stuffing Game, we now characterize a Nash equilibrium of the original (simultaneous move) game. Interestingly, we find that this Nash equilibrium is closely tied to the Stackelberg equilibrium analyzed above. Specifically, the strategy of \st\ remains the same. However, the Nash equilibrium involves a \emph{mixed} strategy on part of~\ins~(since any pure strategy enables a favourable unilateral deviation on part of \st). Moreover, the mixed strategy of~\ins\ involves a \emph{non-uniform} randomization over the set~$A_{z^*}.$

		In what follows, we first define a certain randomized strategy~$q^*$ for \ins~(over $K$-sized subsets of~$A_{z^*}$), and then show that $(z^*, q^*)$ is a Nash equilibrium of the Ballot Stuffing Game.
		
		We begin by defining the probability~$p_a$ that booth~$a \in A_{z^*}$ is chosen for inspection by~\ins:
		$$p_a := 1- \frac{g'_a(z^*_a)}{\theta_{z^*}}.$$ Note that the probability of \emph{not} choosing booth~$a$ for inspection is proportional to~$g'_a(z^*_a),$ and the normalization by the structure slope $\theta_{z^*}$ ensures that (i) $p_a \leq 1$ (recall the \emph{slope order condition} in Theorem~\ref{thm:Stackelberg-structure}), and (ii) $\sum_{a \in A_{z^*}} p_a = K,$ which is consistent with a total of~$K$ inspectors being deployed. The intuition behind $1 - p_{a} \propto g'_{a}(z_a^*)$ is that \ins\ chooses to inspect booths that are `easier to stuff,' i.e., booths with a smaller marginal cost of planned stuffing, with higher probability.

		It remains to show that there exists a probability mass function~$q^*$ over the space $\mathcal{A}_{z^*}^K$ of $K$-sized subsets of~$A_{z^*}$ that is consistent with $(p_a, a \in A_{z^*}),$ i.e., 
		\begin{equation}
			\label{eq:p-q_consistency}
			\displaystyle
			p_a = \sum_{I \in \mathcal{A}_{z^*}^K \colon a \in I} q^*_I.
		\end{equation}
		Here, $q^*_I$ denotes the probability that~\ins\ chooses the set of booths~$I \subset A_{z^*}$ for inspection, where $|I| = K.$ That such a $q^*$ exists is shown below.
		\begin{lemma}
			\label{lemma:q_existence}
			There exists a probability mass function~$q^*$ over the set $\mathcal{A}_{z^*}^K$ satisfying~\eqref{eq:p-q_consistency}.
		\end{lemma}
		The proof of Lemma~\ref{lemma:q_existence}, which uses Farkas' lemma (see~\cite{Boyd}), can be found in the Appendix.
		
		We are now ready to state and prove our main result.
		
		\begin{theorem}[\textbf{Nash Equilibrium Characterization}] 
			\label{thm:ballot_NE}
			A strategy profile $(z^*,q^*),$ where~$z^*$ is the stuffing vector that optimizes~\emph{\ref{eq:Stackelberg2}}, and~$q^*$ is a probability distribution over $K$-sized subsets of~$A_{z^*}$ satisfying~\eqref{eq:p-q_consistency}, is a Nash equilibrium of the Ballot Stuffing Game.
		\end{theorem}
		
		It is interesting to note that the Nash equilibria of the ballot stuffing game identified in Theorem~\ref{thm:ballot_NE} involve:
		\begin{itemize}
			\item no randomization on the part of \st\ (whose action space is continuous),\footnote{Indeed, it can be shown that given any (possibly randomized) strategy of \ins, there exists a non-randomized optimal best response for \st.} but
			\item  a delicate, marginal-cost-aware, randomization on the part of \ins\ (whose action space is discrete).
		\end{itemize}
		Additionally, a computation of the equilibrium of the Stackelberg variant of the Ballot Stuffing Game considered in Section~\ref{sec:Stackelberg}, or equivalently, a computation of the optimal solution of~\ref{eq:Stackelberg2}, enables a computation of the Nash equilibria identified in Theorem~\ref{thm:ballot_NE}; we describe a tractable algorithm for solving~\ref{eq:Stackelberg2} in Section~\ref{sec:algo}. The remainder of this section is devoted to the proof of Theorem~\ref{thm:ballot_NE}.
		
		
		\begin{proof}[Proof of Theorem~\ref{thm:ballot_NE}]
			Consider an action profile~$(z^*,q^*)$ as prescribed in the statement of the theorem. To prove that this action profile constitutes a Nash equilibrium, it suffices to show that neither player has an incentive for unilateral deviation. Given the action~$z^*$ on part of \st, it is easy to see that $q^*$ is a best response on part of \ins, since it places all~$K$ inspectors at booths in $A_{z^*}$ (albeit, via a randomized booth selection).
			
			It therefore suffices to fix the (randomized) action $q^*,$ and show that $z^*$ is a best response for \st, in that it maximizes the \emph{average} number of stuffed votes. 
			\ignore{
				Suppose we run the optimization problem \ref{eq:opt-problem2} and get their optimal values with the corresponding cost functions as 
				\begin{align}
					g_{a_1}(z_a^*),..g_{a_{|A|}}(z_a^*), g_{b_1}(z_{b_1}^*),..g_{b_{|B|}}(z_{b_{|B|}}^*),
					g_{c_1}(0),..g_{c_{|C|}}(0) \nonumber
				\end{align}
				To prove that these strategies form a Nash Equilibrium, we need to fix each strategy and show that changing the other strategy unilaterally cannot result in a better outcome. $P_1$'s strategy is deterministic, so $P_A$ needs to select the maximum $K$ ballots, which it does. Since $P_A$ uses a mixed strategy, $P_1$ needs to maximize the expected number of stuffed votes. This can be formulated as an optimization problem:
			} 
			Given $q^*,$ the best response of \st\ is posed as the following optimization.
			\begin{align}
				\max_{z} \quad & \sum_{a \in A_{z^*}} (1 - p_{a}) z_{a} + \sum_{b \in B_{z^*}} z_{b} + \sum_{c \in C_{z^*}} z_{c} \nn \\
				\text{subject to} \quad & \sum_{j=1}^J g_j(z_j) \leq G 
				\label{opt:stuffer_best_response}
			\end{align}
			Note that the vote stuffing planned in any booth~$a \in A_{z^*}$ materializes with probability~$(1-p_a).$ The optimization problem~\eqref{opt:stuffer_best_response} is similar to \ref{eq:Stackelberg2} in the no-inspector case ($K=0$). Indeed, consider the transformation:
			\begin{align*}
				\hat{z}_{j} &= (1 - p_{j}) z_{j} \quad  \text{ for } j \in A_{z^*},\\
				\hat{z}_{j} &= z_j  \quad \text{ for } j \notin A_{z^*},\\
				\hat{g}_{j}(\hat{z}_{j}) &= g_j\left(\frac{\hat{z}_{j}}{1-p_j} \right) \quad  \text{ for } j \in A_{z^*},\\
				\hat{g}_{j}(\hat{z}_{j}) &= g_{j}(\hat{z}_{j}) \quad \text{ for } j \notin A_{z^*}.
			\end{align*}
			The optimization~\eqref{opt:stuffer_best_response} is therefore equivalent to:
			\begin{align}
				\max_{\hat{z}} \quad & \sum_{j} \hat{z_j} \nn \\
				\text{subject to} \quad & \sum_{j} \hat{g}_j(\hat{z}_j) \leq G       \label{opt:stuffer_best_response_tx}
			\end{align} 
			Showing that $z^*$ is optimal for~\eqref{opt:stuffer_best_response} is thus equivalent to showing that $\hat{z}^*$ is optimal for ~\eqref{opt:stuffer_best_response_tx}. We now prove the latter.
			Note that for $a \in A_{z^*},$ 
			\begin{align}
				\hat{g}'_a(\hat{z}^*_j) &= \frac{1}{1-p_a} g'_a\left(\frac{\hat{z}^*_a}{1-p_a}\right) = \frac{1}{1-p_a} g'_a(z^*_a) \nn \\
				&= \theta_{z^*}. \label{eq:stuffer_best_response_tx_1}
			\end{align}
			Here, we have used the fact that $(1-p_a) \propto g'_a(z^*_a).$ It therefore follows that $$\hat{g}'_a(\hat{z}^*_a) \stackrel{(i)}{=} \hat{g}'_b(\hat{z}^*_b) \stackrel{(ii)}{\leq} \hat{g}'_c(\hat{z}^*_c) \  \forall a \in A_{z^*}, b \in B_{z^*}, c \in C_{z^*}.$$ Here, $(i)$ follows from \eqref{eq:stuffer_best_response_tx_1} and the \emph{structure slope condition} satisfied by $z^*,$ and $(ii)$ follows from the \emph{slope order condition} satisfied by $z^*.$ This proves that $\hat{z}^*$ is optimal for ~\eqref{opt:stuffer_best_response_tx}, which completes the proof.
		\end{proof}
		
		\ignore{
			\begin{lemma}
				$z^*$ is a valid solution to the above optimization problem.
			\end{lemma} 
			We know that the solution to the above problem adheres to the \textbf{Optimal Structure}. Specifically, when \( K = 0 \), \textbf{Theorem 2} decomposes to all the marginal values in sets \( A \) and \( B \) to be equal and less than those in set \( C \) and \textbf{Sum condition}. To verify that \( z^* \) is a valid solution, we need to demonstrate that it satisfies above two conditions.
			
			We start by confirming that the marginals within sets \( A \) and \( B \) are equal:
			\begin{align*}
				\textrm{TPT }\frac{1}{1 - P_{a_1}} g'_{a_1}(z^*_{a_1}) &=  \ldots = \frac{1}{1 - P_{a_{|A|}}} g'_{a_{|A|}}(z^*_{a_{|A|}}) \\
				=\ g'_{b_1}(z^*_{b_1}) = & \ g'_{b_2}(z^*_{b_2}) = \ldots = g'_{b_{|B|}}(z^*_{b_{|B|}})
			\end{align*}
			
			\begin{align*}
				\Leftrightarrow \frac{\sum_{j=1}^{|A|} g'_{a_j}(z^*_a)}{|A| - K} \frac{g'_{a_1}(z^*_a)}{g'_{a_1}(z^*_a)} &= \ldots = \frac{\sum_{j=1}^{|A|} g'_{a_j}(z^*_a)}{|A| - K} \frac{g'_{a_{|A|}}(z^*_a)}{g'_{a_{|A|}}(z^*_a)} \\
				=\ g'_{b_1}(z^*_{b_1}) = & \ g'_{b_2}(z^*_{b_2}) = \ldots = g'_{b_{|B|}}(z^*_{b_{|B|}})
			\end{align*}
			The above is true by the \textbf{Structure Slope condition} of the optimization problem \ref{eq:opt-problem1}.
			
			Furthermore, $g'_{b}(z^*_{b})>g'_{c}(z^*_{c})$ again using the \textbf{Structure Slope condition}. $z^*$ also satisfies the \textbf{Sum condition} because it was a solution of \ref{eq:opt-problem1}. Thus, $z^*$ is indeed the optimal solution.
		} 

		\ignore{
			\begin{remark}
				The strategy of $P_A$ can be formulated as a Linear Program in terms of $P_{a_i}$ whose existence can be proved using this Lemma 4.
			\end{remark}
		} 
		
		\section{Algorithm for computing~$z^*$}
		\label{sec:algo}
		In this section we present an efficient algorithm to solve~\ref{eq:Stackelberg2}, which exploits the structure of its optimal solution (see Theorem~\ref{thm:Stackelberg-structure}). To simplify the exposition, we make the following assumption on the collection of cost functions.
		
		\begin{assumption}[Cost Ordering]
			\label{assumption:mono_constraints}
			There exists an ordering of the booths, such that the marginal stuffing costs are ordered as follows:
			\begin{align}
				\label{eq:ordering1}
				g_1'(z) \leq g_2'(z) \leq g_3'(z) \leq \cdots \leq g_J'(z) \ \forall\ z
			\end{align}
			Additionally, the marginal stuffing cost at zero is zero, i.e., 
			\begin{align}
				\label{eq:ordering2}
				g_1'(0) = g_2'(0) = g_3'(0)  = \ldots = g_J'(0)=0.
			\end{align}
		\end{assumption}
		The condition~\eqref{eq:ordering1} states that the cost function marginals do not `intersect.' Under this ordering, note that for booths~$i$ and $j,$ where $i < j,$ booth~$i$ is `cheaper' to stuff than booth~$j.$ The condition~\eqref{eq:ordering2} implies that the equilibrium stuffing vector~$z^*$ is strictly positive on all coordinates, i.e., $C_{z^*} = \emptyset.$ Additionally, Assumption~\ref{assumption:mono_constraints} implies that the sets~$A_{z^*}$ and $B_{z^*}$ have the following structure.
		
		\begin{lemma}\label{lemma:A-Bstructure}
			Under Assumption~\ref{assumption:mono_constraints}, there exists $\ell$ satisfying $K+1 \leq \ell \leq J,$ such that $A_{z^*} = \{1,\ldots,\ell\},$ and $B_{z^*} = \{\ell+1,\ldots,J\}.$
		\end{lemma}

		
		Lemma~\ref{lemma:A-Bstructure} implies that from an algorithmic standpoint, the search for $A_{z^*}$ is restricted to $J - K$ possibilities. While it is possible to perform an exhaustive search over these possibilities to determine $z^*$ (as shown below), it turns out that one can further exploit the structure of $z^*$ to make the search for $A_{z^*}$ more efficient.
		At the heart of the proposed algorithm is the routine \calcZ, which is stated formally as Algorithm~\ref{alg:calc_z}. This routine takes as input a candidate choice of $|A|,$ and returns a candidate~stuffing vector $\mathbf{z}$ along with a certain status flag.
		\begin{algorithm}[H]
			\caption{Compute candidate $\mathbf{z}$ and give error status}
			\label{alg:calc_z}
			\textbf{Function: \calcZ}\\
			\textbf{Initialize:} $g_i(z_i), G, K$\\
			\textbf{Input:} $|A|$ \\
			\textbf{Output:} $\mathbf{z}, \text{status}$ 
			\begin{algorithmic}
				\STATE Solve \eqref{eq:Solve_za} for $z_a$
				\STATE Compute $\mathbf{z}$ using \eqref{eq:candidate_z}
				\IF{$z_{|A|+1}>z_{a}$}
				\STATE \textbf{status} $\gets$ Error~P
				\ELSIF{$g'_{|A|}(z_{a})> \frac{\sum_{i=1}^{|A|} g'_{i}(z_a)}{|A|-K}$}  
				\STATE \textbf{status} $\gets$ Error~Q
				\ELSE
				\STATE \textbf{status} $\gets$ correct
				\ENDIF
				\STATE \textbf{Return} $\mathbf{z}, \text{status}$
			\end{algorithmic}
		\end{algorithm}
		Specifically, this routine first computes the unique positive scalar~$z_a$ satisfying 
		\begin{align}
			\sum_{i=1}^{|A|} g_{i}(z_a) + \sum_{i=|A|+1}^{J} g_{b}((g'_{i})^{-1}({\theta(z_a)})) = G,
			\label{eq:Solve_za}
		\end{align}
		where $$\theta(z_a) := \frac{\sum_{i=1}^{|A|} g'_{i}(z_a)}{|A|-K}.$$ It then uses $z_a$ to construct the candidate stuffing vector~$\mathbf{z}$ defined as: 
		\begin{equation}
			\label{eq:candidate_z}
			\mathbf{z}_i = \left\{
			\begin{array}{cl}
				z_a & \text{ for } 1 \leq i \leq |A|, \\
				(g'_{i})^{-1}({\theta(z_a)}) & \text{ for } i > |A|.
			\end{array}
			\right.
		\end{equation}
		Finally, \calcZ\ checks whether this vector $\mathbf{z}$ satisfies the following conditions:
		\begin{enumerate}
			\item $z_a \geq \mathbf{z}_{|A|+1},$
			\item $g'_{|A|}(z_a) \leq \theta(z_a).$
		\end{enumerate}
		\begin{remark}
			When $|A| = J$, condition 1 is undefined, and when $|A| = K$, condition 2 is undefined. We ensure that the algorithm avoids evaluating these conditions in these corner cases.
		\end{remark}
%
		
		\begin{algorithm}
			\caption{Algorithm for solving \ref{eq:Stackelberg2}}
			\label{alg:monocon}
			\textbf{Input:} $g_i(\cdot), G, J, K$ \\
			\textbf{Output:} $z^*$ \\
			\textbf{Variables:} $x,y,|A|$ \\
			\textbf{Initialize:} $x=K+1,y=J$
			\begin{algorithmic}[1]
				\WHILE{True}
				\IF{$|x-y|\leq 1$}
				\STATE $\mathbf{z}$, status $\gets$ $\text{\calcZ}(y)$ 
				\IF{status = \textbf{correct}}
				\STATE Return $\mathbf{z}$
				\STATE \textbf{Break}
				\ELSE
				\STATE $\mathbf{z}$, status $\gets$ $\text{\calcZ}(x)$ 
				\STATE Return $\mathbf{z}$
				\STATE \textbf{Break}
				\ENDIF
				\ELSE
				\STATE $|A| \gets \lceil{\frac{x+y}{2}}\rceil$
				\STATE $\mathbf{z}, \text{status} \gets \text{\calcZ}(|A|)$
				\IF{$\text{status} = \text{correct}$}
				\STATE \textbf{Return} $\mathbf{z}$
				\STATE \textbf{break}
				\ELSIF{$\text{status} = \text{Error P} $}
				\STATE $x \gets |A|$
				\ELSIF{$\text{status} = \text{Error Q}$}
				\STATE $y \gets |A|$
				\ENDIF
				\ENDIF
				\ENDWHILE
			\end{algorithmic}
		\end{algorithm}	
		If Conditions~1 and 2 above both hold, it follows from Theorem~\ref{thm:Stackelberg-structure} that $\mathbf{z}$ is optimal. In this case, \calcZ\ returns the status `correct.' Else, if Condition~1 is violated, it returns the status `Error~P,' and if Condition~2 is violated, it returns the status `Error~Q.' (It is easy to show that both Conditions~1 and~2 cannot be violated.) 
		
		As such, note that one can solve \ref{eq:Stackelberg2} by simply querying the routine \calcZ\ for the different possibilities of $|A|,$ ranging from $K+1$ to $J.$ However, it is possible to exploit the tag of the error status (i.e., P or Q) to check whether the queried value of $|A|$ is larger, or smaller, than $|A_{z^*}|.$
		\begin{lemma}
			\label{lemma: PQerrors}
			Under Assumption~\ref{assumption:mono_constraints}, 
			\begin{itemize}
				\item if \calcZ\ returns`Error~P,' then $|A| < |A_{z^*}|,$
				\item if \calcZ\ returns`Error~Q,' then $|A| > |A_{z^*}|.$
			\end{itemize}
		\end{lemma}
		Lemma~\ref{lemma: PQerrors} enables a binary search for the correct value of~$|A|$ over the range $K+1$ to $J;$ this requires at most $\log_2(J - K)$ queries of \calcZ. The corresponding algorithm, which provably solves~\ref{eq:Stackelberg2} under Assumption~\ref{assumption:mono_constraints}, is stated formally as Algorithm~\ref{alg:monocon}.

		\section{Parliamentary Model}
		\label{sec:Parliamentary}
		In general, the Parliamentary model is harder to solve because because the different booths can have different weights, and different probabilities of winning. Clearly, it would be more beneficial to stuff in close contests rather than in `surer' ones. In the following, we present an approximation to transform it to a version of the Plebiscite model that we considered above.  Specifically, consider the problem of $P_1$ that is written as
		\begin{align*}
			&\max_{z} \ \phi(z,y)  = \prob{ \ \left( \sum_{j=1}^J \Ind
				{ (X_j+(1-y_j)z_j) > 0 } \right) > J/2 }.
		\end{align*}
		Under a Markov inequality relaxation we can write this as 
		maximization as being approximately the same 
		\begin{align*}
			& \max_z \EXP{\sum_{j=1}^J \Ind
				{ (X_j+(1-y_j)z_j) > 0 }},
		\end{align*}  
		which is the same as 
		\begin{align*}
			\max_z \sum_{j=1}^J \prob{(X_j+(1-y_j)z_j) > 0}
		\end{align*}
		The probability terms in the summation is a tail distribution of a Gaussian which we can approximate by function $f_j$ for  station $j.$ Assume $f_j$ is a concave increasing function with $f_j(0)=0$. Define $w_j:=f_j(z_j)$. Now our objective function becomes
		\begin{align*}
			&\max_z \sum_{j=1}^J f_j(z_j)(1-y_j)\\
			\Leftrightarrow &\max_z \sum_{j=1}^J w_j(1-y_j)
		\end{align*}	
		So now our Parliamentary model can be represented as a Plebiscite model in terms of $w_j$ as our stuffing votes and $g_j(f_j^{-1})$ as our new cost functions.
		\begin{align}
			& \max_{z} \ & \quad & \ \sum_{j=1}^J w_j(1-y_j) \nonumber \\
			& \mbox{subject to}  && \sum_{j=1}^J g_j(f_j^{-1}(w_j)) \ \leq \ G, \nonumber \\
			& && \sum_{j=1}^{J} y_j=K,
			\label{eq:Parl-to-Pleb}
		\end{align}
		
		\section{Numerical Experiments}
		\label{sec:Numerical}
		We now provide some numerical illustrations. First consider the case of $J=5$ and a budget $G.$ Let the cost functions for stuffing in the five booths be $z_1^4,$ $z_2^4,$ $2z_3^4,$ $5z_4^2,$ and  $15z_5^2.$ We obtain the stuffing vector for different $G$ and $K.$ 
		Since the cost functions are `intersecting', we cannot use the algorithms of Section~\ref{sec:algo}; we use CPLEX to solve \ref{eq:Stackelberg2}. For $K=0,$ the stuffing vector $z^{*}$ is plotted as a function of the budget $G$ in Fig.~\ref{fig:plot_K0}. For $K=1$ and $K=2,$ in Figs.~\ref{fig:plot_K1} and \ref{fig:plot_K2} we plot the equilibrium values of the number of votes stuffed in the five booths as a function of G.
		\begin{figure}[H]
			\centering
			\includegraphics[width=0.48\textwidth]{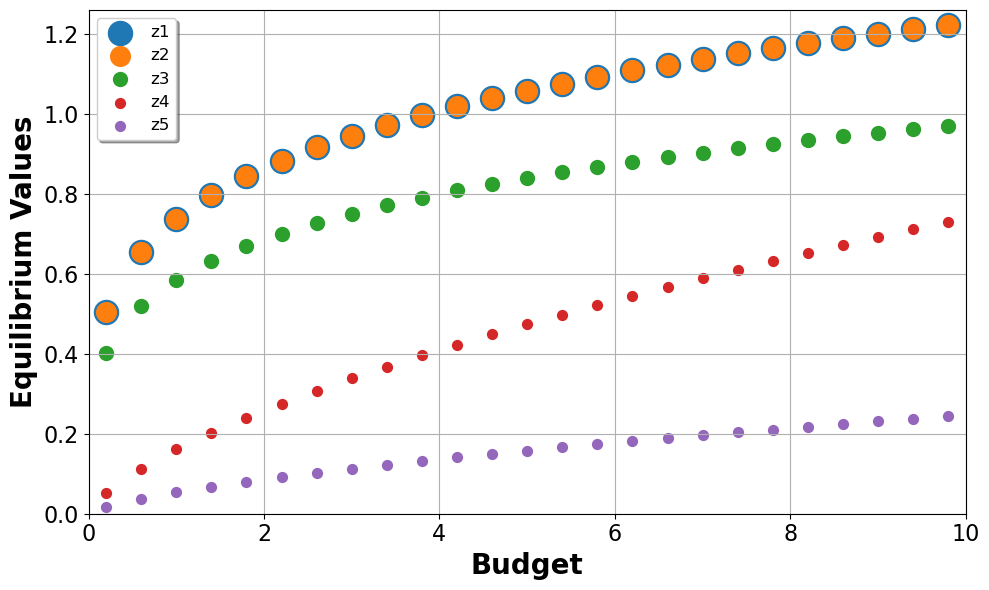}
			\caption{Optimum stuffed votes $z^*$ as a function of budget $G$ for $K=0$ inspectors.}
			\label{fig:plot_K0}
		\end{figure}
		\begin{figure}[H]
			\centering
			\includegraphics[width=0.48\textwidth]{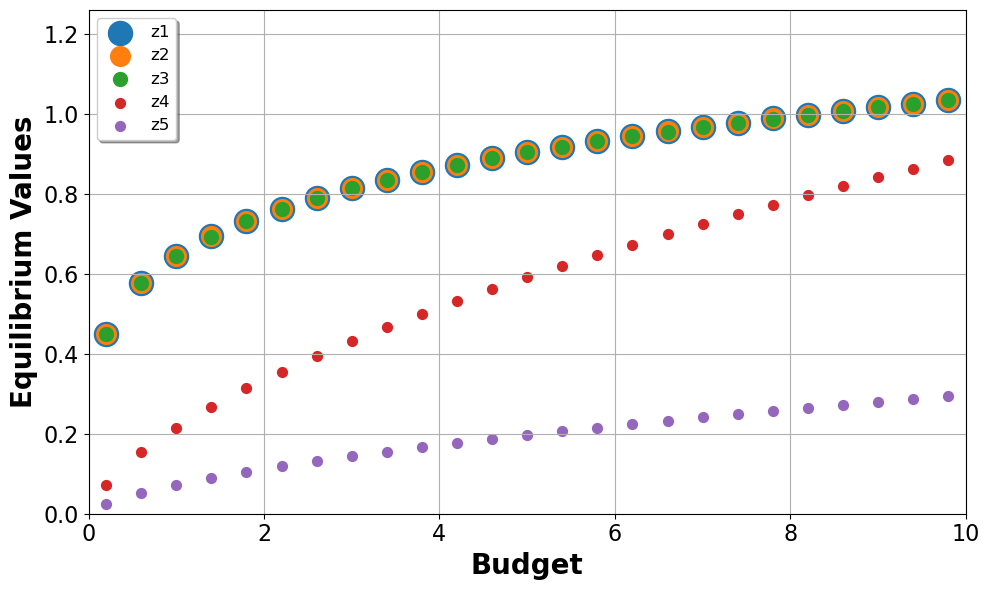}
			\caption{Equilibrium stuffed votes $z^*$ as a function of budget $G$ for $K=1$ inspector.}
			\label{fig:plot_K1}
		\end{figure}
		\begin{figure}[H]
			\centering
			\includegraphics[width=0.48\textwidth]{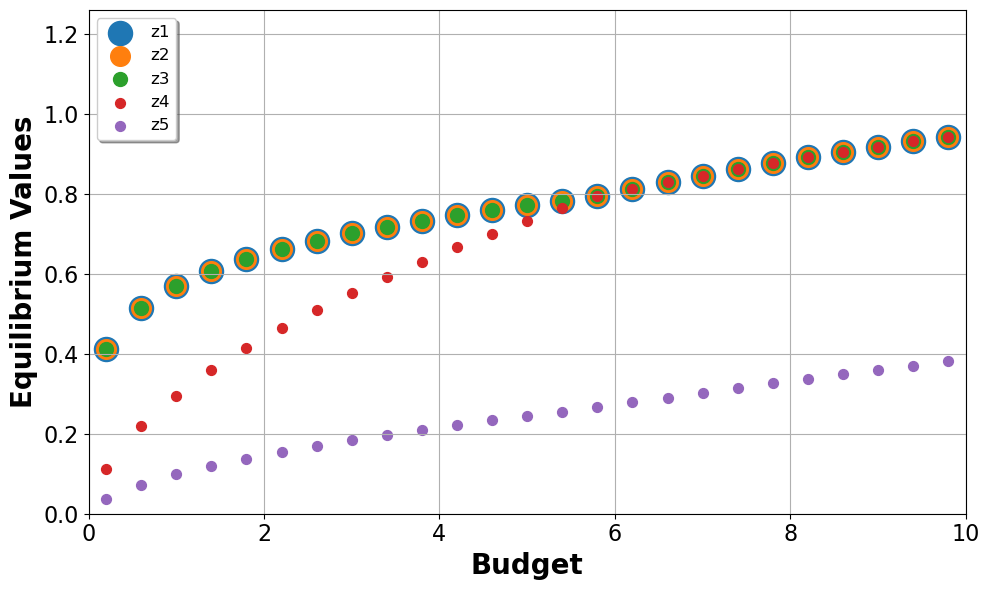}
			\caption{Equilibrium stuffed votes $z^*$ as a function of budget $G$ for $K=2$ inspectors.}    
			\label{fig:plot_K2}
		\end{figure}
		Observe the significant change in the stuffing vector as $K$ increases from 0. Specifically, the $z^*$ values in set $A_z $ decrease while those in set $B_z$ increase---since the inspectors will negate the votes in booths of set $A_z$, $P_1$  increases stuffing in set $B_z.$ From the perspective of the structure slope ($\theta_z$), as the number of inspectors grows, $\theta_z$ also rises. However, the slopes for set $B_z$ remain constant, leading to higher $z^*$ values in set $B_z$ at the same budget. Also observe that in Fig. \ref{fig:plot_K2}, $z_4$ transitions from set $B_z$ to set $A_z.$ This happened because the $z_4$ caught up with the $z_a$ values as the budget was increased. \\
		
		We now consider a slightly more elaborate example of a nationwide plebiscite with the fifty one states as the fifty one battlefields and two parties in contention. We assume that $X_i$ is the random number of excess votes that will be cast for $P_1$ in state $i.$ Using the state-level election data from the previous twelve presidential elections we estimate the mean and variance of $X_i$. The votes cast for the Republicans and Democrats is normalised to the population of 2020.
		The sample mean and sample variance of the normalised vote count is used and we assume that if a plebiscite were to be held today,
		the two parties would get votes that are Gaussian.
		\begin{table}[H]
			\centering
			\scriptsize 
			\begin{subtable}{0.45\linewidth}
				\centering
				\caption{G = 10,000}
				\label{tab:g_10000}
				\begin{tabular}{|c|c|c|}
					\hline
					\textbf{K} & \textbf{Stuffed votes} & \textbf{Prob of win} \\
					\hline
					2  & 63032  & 0.423515 \\
					\hline
					5  & 54758  & 0.422737 \\
					\hline
					8  & 48541  & 0.422152 \\
					\hline
					10 & 44898  & 0.421810 \\
					\hline
				\end{tabular}
			\end{subtable}
			\hfill
			\begin{subtable}{0.45\linewidth}
				\centering
				\caption{G = 100,000}
				\label{tab:g_100000}
				\begin{tabular}{|c|c|c|}
					\hline
					\textbf{K} & \textbf{Stuffed Votes} & \textbf{Prob of win} \\
					\hline
					2 & 199325  & 0.436372 \\
					\hline
					5 & 173161  & 0.433898 \\
					\hline
					8 & 153503  & 0.432041 \\
					\hline
					10 & 141981 & 0.430954 \\
					\hline
				\end{tabular}
			\end{subtable}
			\\
			\begin{subtable}{0.45\linewidth}
				\centering
				\caption{G = 1,000,000}
				\label{tab:g_1000000}
				\begin{tabular}{|c|c|c|}
					\hline
					\textbf{K} & \textbf{Stuffed Votes} & \textbf{Prob of win} \\
					\hline
					2 & 630322 & 0.477412 \\
					\hline
					5 & 547585 & 0.469501 \\
					\hline
					8 & 485419 & 0.463565 \\
					\hline
					10 & 448984 & 0.460090 \\
					\hline
				\end{tabular}
			\end{subtable}
			\hfill
			\begin{subtable}{0.45\linewidth}
				\centering
				\caption{G = 10,000,000}
				\label{tab:g_10000000}
				\begin{tabular}{|c|c|c|}
					\hline
					\textbf{K} & \textbf{Stuffed Votes} & \textbf{Prob of win} \\
					\hline
					2 & 1993254 & 0.606694 \\
					\hline
					5 & 1731616 & 0.582335 \\
					\hline
					8 & 1535032 & 0.563817 \\
					\hline
					10 & 1419813 & 0.552895 \\
					\hline
				\end{tabular}
			\end{subtable}
			\caption{Table shows the number of effectively stuffed votes and the probability of $P_1$ winning the plebiscite with different $G$ and $K.$
				\label{table:Expt1}
			}
		\end{table}
		With this, the probability of $P_1$ winning the nationwide plebiscite is $0.4176.$ Let $\mu_i$ and $\sigma_i$ be, respectively, the mean and standard deviation of the number of votes that $P_1$ (arbitrarily chosen as the Democratic party) would get in state $i$ if a plebiscite were to be held.
		
		We assume that $P_1$ intends to stuff and has a cost function \( g_i(z_i) = \frac{z_i^2}{\sigma_i^2 N_i}; \) we reason that it is easier to stuff in a state with higher population and higher variance. 
		
		In Table~\ref{table:Expt1}, we show the number of stuffed votes that were not negated by the inspectors and the new probability of $P_1$ winning the plebiscite for different values of $G$ and $K.$ 

		Next we fix  $G=8,000,000$ and vary $K$ and calculate the probability of $P_1$ winning as $K$ is increased. This probability is plotted in Figure~\ref{fig:exp2}. 
		
		\begin{figure}[H]
			\centering
			\includegraphics[width=0.45\textwidth]{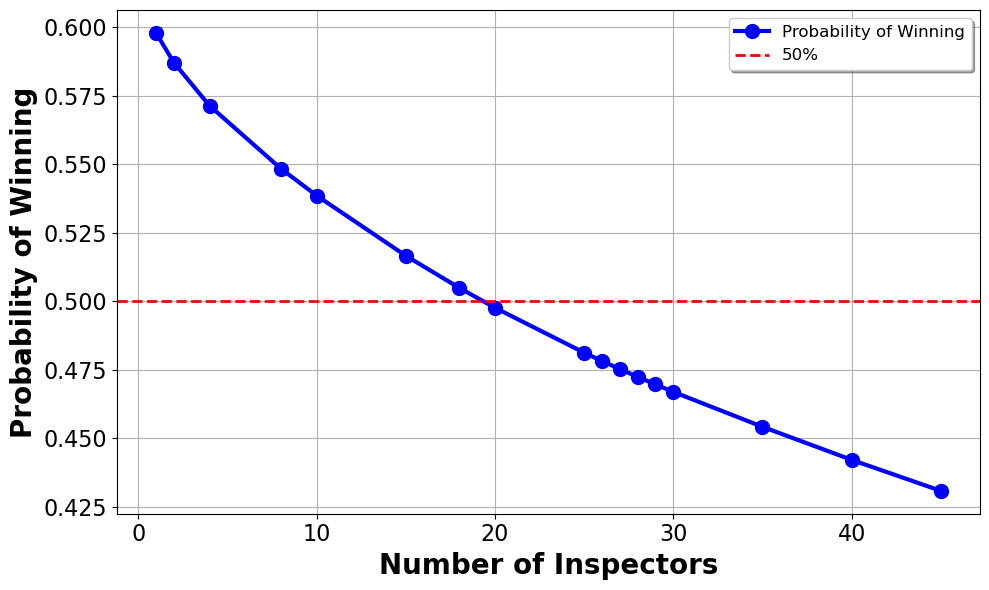}
			\caption{Probability of $P_1$ winning the nationwide plebiscite as a function of $K$ when $G=8,000,000.$}
			\label{fig:exp2}
		\end{figure}

		\section{Discussion and Conclusion}
		In this paper we defined the Ballot Stuffing Game variation of the Colonel Blotto Game. For a Plebiscite election, the structure of the Nash equilibrium vector of the `stuffing votes' and the corresponding deployment of the inspectors were derived. An approximation to the relaxed version of the Parliamentary election was shown to be the same as the Plebiscite model. We remark here that for the Plebiscite model, the case when the amount that can be stuffed in booth $i$ is capped, say at $\overline{z}_i,$ has a similar solution for both the structure of the Equilibrium and the algorithm to compute the same; we have elided that discussion to simplify the exposition. 
		
		Several options for future work present themselves. A more realistic version of the Parliamentary model where $P_1$ stuffs to increase  its  the probability of winning while $P_A$ deploys inspector to minimise this increased probability is an obvious extension. An even more realistic case would be a three-party game where both the contestants  strategically indulge in stuffing even as the inspectors by the Election Commission to minimise the change in the outcome probabilities caused by stuffing. Here it is possible that $P_A$ `allows' some stuffing that perhaps cancel each other out. 
		\bibliographystyle{aaai25}
		\bibliography{Blotto}
		\begin{appendix}
			
			\section{Technical Appendix}
			
			\subsection{Proof of Lemma~\ref{lemma:Stackelber1}}
			We will first demonstrate that the cost constraint holds in equality, $\sum_{j=1}^{J} g_j(z_j) = G.$ Suppose it does not, then by just increasing the minimum $z^*$ value, $z_{(J)}^*$. We will  increase the objective function, contradicting the fact that maximum was attained before the adjustment.\\
			
			Using the above we are in the shape to prove uniqueness for \ref{eq:Stackelberg1}. To the contrary, assume we have two distinct solution $z^*,\tilde{z}.$	Note that our optimization problem is convex, so any point lying on the segment in between $z^*,\tilde{z}$ will also be a solution. Let $\hat{z}=\frac{z^*+\tilde{z}}{2}$, which is also a solution of \ref{eq:Stackelberg1}, hence will also satisfy $\sum_{j=1}^{J} g_j(\hat{z}_j) = G.$ Using strict convexity of $g_j()$, we have $\sum_{j=1}^{J} g_j(\hat{z}_j)< \sum_{j=1}^{J} \frac{g_j(z^*_j)}{2}+\sum_{j=1}^{J} \frac{g_j(\tilde{z}_j)}{2}=G$ proving the contradiction.\\
			
			Now we prove that set of indices of $z^*$ that have the maximum value has cardinality at least $K+1$. Note the objective of the optimization problem is to maximize the total number of votes in the minimum $(J-K)$ ballots. Hence the largest $K$ values will be kept as small as possible by the optimization problem, that is indeed equal to $z^*_{(K+1)}$ making the cardinality strictly more than $K$.
			\subsection{Proof of Lemma \ref{lemma:q_existence}}
			We will restate our Lemma \ref{lemma:q_existence} in terms of matrix equations.
			\begin{lemma}
				\label{lemma:exist-LM}
				Consider the following system of equations: $R q = p$, where $R \in \mathbb{R}^{n \times m}$, $q \in \mathbb{R}^m$, and $p \in \mathbb{R}^n$ with $m=\binom{n}{v}$. If the components of the vector $p$ lie between $0$ and $1$ and their sum is equal to $v$, and if the matrix $R$ is defined such that its columns are from the set $S_R:=\{r|\sum_{j=1}^{n}r_j=v,r_j \in \{0,1\} \}$, then there exists a solution to this system.
			\end{lemma}

			To prove the following lemma, we use the following Farkas lemma:~\cite{Boyd}\\\\
			Let $\mathbf{R} \in \mathbb{R}^{n\times m}$ and
			$\mathbf{p} \in \mathbb{R}^{n}$. Then exactly one of the following two
			assertions is true:
			\begin{enumerate}
				\item There exists an $\mathbf{q} \in \mathbb{R}^{m}$ such
				that $\mathbf{Rq} = \mathbf{p}$ and $\mathbf{q} \geq 0$.
				\item There exists a $\mathbf{y} \in \mathbb{R}^{n}$ such that
				$\mathbf{R}^{\mathsf{T}}\mathbf{y} \geq 0$ and
				$\mathbf{p}^{\mathsf{T}}\mathbf{y} < 0$.
			\end{enumerate}
			To prove our claim by contradiction, lets assume second statement is true and there exists such a $y$. The	matrix $\mathbf{R}^{\mathsf{T}}$ is of dimensions $(m \times n)$, where each row represents a permutation of $v$ ones and $n-v$ zeros. The inequality $\mathbf{R}^{\mathsf{T}}\mathbf{y} \geq 0$ implies that $r^T.y \geq 0 \ \forall r \in S_R.$ This condition implies any sum of $v-$size subset of $y$ is greater than equal to zero. We also have $\sum_{j=1}^{n}p_i=v$. Let $S := \sum\limits_{i=1}^n p_i y_i$ which is strictly less than zero. Without loss of generality, we assume that $y_i$ is arranged in descending order ($y_1 \geq y_2 \geq \ldots \geq y_n$).
			\begin{claim}
				$\ \ \ S \geq \sum\limits_{i=n-v+1}^n y_i$	
			\end{claim}
			We have $\sum\limits_{i=1}^{n} p_i = v$
			\begin{align*}
				&\Rightarrow \sum\limits_{i=1}^{n-v} p_i = v-\sum\limits_{i=n-v+1}^{n} p_i\\
				&\hspace{1.43cm} =(1-p_{n-v+1})+(1-p_{n-v+2})+ \ldots (1-p_{n})\\
			\end{align*}
			
			\begin{align*}
				&\Rightarrow y_{n-v+1}\sum\limits_{i=1}^{n-v} p_i =y_{n-v+1}(1-p_{n-v+1})+\\
				&\hspace{3cm} \ldots + y_{n-v+1}(1-p_n)\\
				&\because \text{ we assumed } y_1 \geq y_2 \geq \ldots \geq y_n\\ 			
				&\Rightarrow \sum\limits_{i=1}^{n-v} y_i p_i\geq y_{n-v+1}(1-p_{n-v+1})+y_{n-v+2}(1-p_{n-v+2})\\
				&\hspace{2cm}\ldots + y_{n}(1-p_n) \\
				&\Rightarrow \sum\limits_{i=1}^n y_i p_i \geq \sum\limits_{i=n-v+1}^n y_i\\
				&\text{Hence } S \geq \sum\limits_{i=n-v+1}^n y_i \text{ which is a sum of }v-\text{size subset of }y\\
				& \Rightarrow S \geq 0
			\end{align*}
			But our statement says that $S<0$ which is a contradiction. Hence statement 1 must be true, proving our existence of solution.
			\subsection{Proof of Theorem \ref{thm:Stackelberg-structure}}
			We will begin by establishing the sufficiency condition of the Theorem 1. Let us assume there exists a solution $z \in \R_+^{J}$ which satisfies all the structural properties from Theorem \ref{thm:ballot_NE}. Our goal is to show that this is essentially the optimal solution $z^*$ of \ref{eq:Stackelberg2}. To achieve this we need to identify the set of Lagrange multipliers that satisfy all the KKT conditions of \ref{eq:Stackelberg2}. Therefore, the focus of the proof shifts to showing the existence of such Lagrange multipliers. For notational convenience write $1-y$ as $t$. Define $S_t:= \left\{t: \sum_{j} t_j = J-K\right\}$.
			
			For \ref{eq:Stackelberg2}, we now write the Lagrangian
			$\mathcal{L}:\mathcal{R}^J \times \mathcal{R} \times \mathcal{R}^m
			\rightarrow \mathcal{R}$ with Lagrangian multipliers $\beta_i, \gamma, \delta_t \forall \ i \in \{1,\ldots,J\}
			\text{ and } t \in S_t$ with cardinality $m$ as
			\begin{equation}
				\begin{split}
					\mathcal{L}(z,t,\beta_i, \gamma, \delta_t) = -U &+ \sum_{i=1}^{J}\beta_i(-z_i) + \gamma\left(\sum_{i=1}^{J} g_i(z_i) - G\right) \nonumber \\
					&+ \sum_{t \in S_t} \delta_t\left(U - z.t\right) \\
				\end{split}
				\label{eq:lagrangian}
			\end{equation}
			
			The KKT conditions for our optimization problem is:
			\begin{enumerate}
				\item \label{cond:optimality}$ \beta_i, \gamma, \delta_t \geq 0$ for $1 \leq i \leq J$  and $t \in S_t$ 
				\item $\beta_i z_i = 0 $ for $1 \leq i \leq J$
				\item $\delta_t(U-z.t) \text{ for } t \in S_t$ 
				\item $\gamma(\sum\limits_i^J g_i(z_i) - G) = 0$
				\item $\xi_i+\beta_i = \gamma g'_i(z_i)\forall i \in \{1, 2, \ldots, J\}$ \\where $\xi_i = \sum_{t\in \{S_t|t_i=1\}} \delta_t$ 
				\item $\sum_{t \in S_t}\delta_t = 1$
			\end{enumerate}
			Lets look at what value will be taken by $U$; it is equal to the minimum sum of $(J-K)$-size subset of $z$. Using the set definitions and Lemma 1, we can say that $U=(J-K)z_a+\sum_{b \in B_z}z_b+\sum_{c \in C_z}z_c$. This expression is because we have the value of \(U\) as the total sum minus the maximum \(K\) ballots, but those \(K\) ballots will all be from set \(A\) as it contains the highest values and has cardinality greater than \(K\). From this, we can comment on the $\delta_t$ terms by looking at the 3rd KKT condition.
			We have $\delta_t=0$ if $t_i=0$ for any $i \in B_z,C_z$. Therefore, $\xi_{b_i} = \xi_{c_i} = \sum_{t \in S_t}\delta_t$, which equals 1 by 6th KKT condition.\\
			So, we conclude that \(\xi_{b_i} = \xi_{c_i} = 1\).
			\\\\
			For $b \in B_z$, we have $\beta_b = 0$ using the 2nd KKT condition. Using the 5th KKT condition, we have:
			\[
			\xi_{b} = \gamma g'_b(z_b)
			\]
			Given that $\xi_b = 1$, it follows that:
			\[
			g'_b(z_b) = \frac{1}{\gamma}
			\]
			Hence, we can determine the value of $\gamma$, which is positive since 
			\[
			g'_b(z_b)> 0.
			\]
			\begin{remark}
				The set $B_z$ may be empty in certain cases. In that case make $\frac{1}{\gamma}=\theta_z$. We do this because $\theta_z=g'_b(z_b) \ \forall \ b \in B_z$ due to \emph{Structure Slope Condition} when $B_z$ is non empty.
			\end{remark}
			For $c \in C_z$, using 5th KKT condition we have
			\begin{align*}
				\xi_c+\beta_c &= \gamma  g'_c(z_c)\\
				\Rightarrow 1+\beta_c &= \gamma g'_c(z_c)\\
				\Rightarrow \beta_c &= \gamma g'_c(z_c) - 1\\
				\text{TPT } \beta_c &\geq 0\\
				\Leftrightarrow &g'_c(z_c) \geq \frac{1}{\gamma} 
			\end{align*}
			Which is true because $\frac{1}{\gamma}=\theta_z \leq g'_c(z_c)$, from our \emph{Structure Slope Condition.}
			\\\\
			For $a \in A_z$, we have $\beta_a=0$ using the 2nd KKT condition. Using 5th KKT condition we have: $\xi_{a}=\gamma g'_a(z_a)$\\	
			
			The only thing left to prove is the existence of non-negative $\delta_t$. Recall that $\delta_t = 0$ if $t_i = 0$ for any $i \in B_z, C_z$. So, we only need to find the values of $\delta_t$ for which $t_i = 1 \ \forall \ i \in B, C$. These remaining $\delta_t$ terms total $\binom{|A_z|}{K}$ in number. Note that we can formulate this as a problem of system of equations using the previous conditions and prove its existence from Lemma \ref{lemma:exist-LM}. We again list the conditions we get from before.
			\begin{enumerate}
				\item $\xi_{a} = \sum_{\{t|t_{a}=1, t_i = 1 \ \forall \ i \in B_z, C_z\}} \delta_t \ \forall a \in A_z$ 
				\item $\xi_{a}=\gamma g'_a(z_a) \forall a \in A_z$
				\item $\sum_{a \in A_z} \xi_{a} = |A_z|-K$
			\end{enumerate}
			It is easy to see existence of remaining $\delta_t$ terms, by writing above conditions as system of equation $Rq=p$. Here $q$ represents the $\delta_t$ terms and $p$ represents to $1-\xi_a$ terms.\\
			
			Now we prove the necessity part of Theorem \ref{thm:Stackelberg-structure}. Here we have the primal dual optimal pair($z^*,\lambda^*$) where $\lambda^*$ denotes the Lagrange multipliers $\beta_i, \gamma, \delta_t \forall \ i \in 1,\ldots,J
			\text{ and } t \in S_t$, which satisfy all the KKT conditions. Remember the values of few Lagrange multipliers which are derived in the previous equation which will be used ahead. We need to prove that this $z^*$ satisfies the optimal structure. \\
			
			The \emph{sum condition} of Theorem \ref{thm:Stackelberg-structure} is already above proved in the appendix. We will now prove the \emph{slope order condition}. \\
			For $c \in C_z$, using 5th KKT condition we have
			\begin{align*}
				\xi_c+\beta_c &= \gamma  g'_c(z_c)\\
				\Rightarrow 1+\beta_c &= \gamma g'_c(z_c) \ \because \xi_c=1\\
				\Rightarrow g'_c(z_c) &= \frac{1+\beta_c}{\gamma}\\
				\Rightarrow  g'_c(z_c) &\geq \frac{1}{\gamma} \ \because \beta_c \geq 0\\
				\Rightarrow  g'_c(z_c) &\geq \theta_z\\			
			\end{align*}
			
			For $a \in A_z$, we have $\beta_a=0$ using the 2nd KKT condition. Using 5th KKT condition we have
			\begin{align*}
				\xi_{a}&=\gamma g'_a(z_a)\\
				\Rightarrow g'_a(z_a)&=\frac{\xi_{a}}{\gamma}\\
				\Rightarrow g'_a(z_a)&\leq \frac{1}{\gamma} \ \because (\xi_{a}<1)\\
				\Rightarrow g'_a(z_a)&\leq \theta_z
			\end{align*}
			
			To prove the \emph{structure slope condition} we manipulate the $\xi_a$ terms using the 5th KKT condition. We have
			\begin{align*}
				&\xi_{a_1}=\gamma g'_{a_1}(z_a)\\
				&\xi_{a_2}=\gamma g'_{a_2}(z_a)\\
				&.\\
				&.\\
				&.\\
				&\xi_{a_{|A|}}=\gamma g'_{a_{|A|}}(z_a)\\
			\end{align*}
			Add all these $\xi_a$ terms together to get:\\
			\begin{align*}
				\sum_{i=1}^{|A|} \xi_{a_i} &= \gamma \sum_{i=1}^{|A|} g'_{a_i}(z_a) \\
				\Rightarrow(|A|-K)\sum_{t \in S_t} \delta_t &= \gamma \sum_{i=1}^{|A|}g'_{a_i}(z_a) \\
				\Rightarrow\frac{\sum_{i=1}^{|A|} g'_{a_i}(z_a)}{(|A|-K)} &= \frac{1}{\gamma} \\
				\Rightarrow\frac{\sum_{i=1}^{|A|} g'_{a_i}(z_a)}{(|A|-K)} &= \theta_z
			\end{align*}
			
			\subsection{Proof of Lemma \ref{lemma: PQerrors}}
			We prove a Lemma which directly proves our result.
			\begin{lemma}
				\label{lemma:error-struc}
				Errors occur in block structure, with block of error P, then our correct $|A_{z^*}|$, then a block of error Q.
			\end{lemma}
			To justify Lemma \ref{lemma:error-struc} we prove a series of Lemmas.
			\begin{lemma}
				When $|A|=K+1$, it will give us the correct solution or Error
				P. When $|A|=J$, it will give us the correct solution or Error Q.
			\end{lemma}
			\begin{proof}
				When $|A| = K + 1$, $\theta_z = \sum_{a=1}^{|A|}
				g'_{a}(z_a)$, implying that $g'_{a}(z_a) < \theta_z$ for $1\leq a \leq |A|$ hence violation Q cannot occur. Conversely, when $|A| = J$,
				all $z$ values are equal hence violation P cannot occur.
			\end{proof}		
			In the following lemmas we see what violations can occur in adjacent $|A|$. By analyzing what errors can occur next to each other we can conclude the block structure.\\
			
			We define a notation to look at the these two adjacent candidate $|A|,|A|+1$. Let us denote $|A|$ by $l$. Let $\mathbf{z}$ values calculated for $l$ be $z_{1,1}, z_{1,2}, \ldots, z_{1,J}$ and for $l+1$ be $z_{2,1}, z_{2,2}, \ldots, z_{2,J}$. Let $z_a$ value for $l$ and $l+1$ be represented by $u$ and $v$, respectively, and let the value of $z_{1,l+1}$ be $u_{l+1}$. Let the structure slope for $l$ and $l+1$ be represented by $\theta_1$ and $\theta_2$ respectively.
			\begin{lemma}
				$v$ lies between $u,u_{l+1}$.
			\end{lemma}
			\begin{proof}
				To prove this by contradiction, assume that $v$	does not lie between $u$ and $u_{l+1}$. Hence, it can be either	greater than both $u$ and $u_{l+1}$ or less than both values. If $v$ is greater than both $u$ and $u_{l+1}$, then $\theta_2>\theta_1$. Which implies the values of $z_{2,i}>z_{1,i} \ \forall i \in \{1,\ldots J\}$ due to Structure Slope Condition from Theorem \ref{thm:Stackelberg-structure}. However, this leads to a contradiction because the values of all ballots cannot increase, given that	$\sum_{i=1}^{J}g_i'(z_i)=G$. Same way we can say $v$ cannot be less
				than both $u$ and $u_{l+1}$.
			\end{proof}
			\begin{lemma}
				Violation Q cannot occur right adjacent to P.
			\end{lemma}
			\begin{proof}
				To prove this by contradiction, assume $l$ exhibit error P and $l+1$ exhibit error Q.
				\begin{align*}
					&\text{Using violation equations we have:} \\
					&u<v<u_{l+1}\\
					&\frac{\sum_{i=1}^{l}g_i'(u)}{l-K} = g_{l+1}'(u_{l+1}) \ \ (\text{Structure slope condition})\\
					&\frac{\sum_{i=1}^{l+1}g_i'(v)}{l+1-K} < g_{l+1}'(v) \ \   (\text{Violation condition of error Q})\\
					&\Rightarrow \sum_{i=1}^{l}g_i'(v) < (l-K)g_{l+1}'(v) \\
					&\Rightarrow \frac{\sum_{i=1}^{l}g_i'(v)}{l-K} < g_{l+1}'(v)\\
					&\text{Using } u < v \text{ we get:}\\
					&\frac{\sum_{i=1}^{l}g_i'(u)}{l-K} <
					\frac{\sum_{i=1}^{l}g_i'(v)}{l-K} < g_{l+1}'(v)\\  
					&\Rightarrow g_{l+1}'(u_{l+1}) < g_{l+1}'(v)\\ 
					&\text{But this is a contradiction because } v < u_{l+1}
				\end{align*}
				So we get from above that P cannot directly go to Q, hence we must
				have a solution in between them or Error Q does not occur for any $l$.
			\end{proof}
			\begin{lemma}
				Violation Q occurs to the right of solution.
			\end{lemma}
			\begin{proof}
				Assume $l$ be the correct choice, we need to prove that $l+1$ will be violation Q. Here we have $u>v>u_{l+1}$.			
				\begin{align*}
					&\text{Structure slope of $l$ is } \frac{\sum_{i=1}^{l}g_i'(u)}{l-K} = g_{l+1}'(u_{l+1})\\
					&\text{Structure slope of $l+1$ is } \frac{\sum_{i=1}^{l+1}g_i'(v)}{l+1-K}\\
					&\text{TPT } g'_{l+1}(v) > \frac{\sum_{i=1}^{l+1}g_i'(v)}{l+1-K}\\
					&\Leftrightarrow g'_{l+1}(v) > \frac{\sum_{i=1}^{l}g_i'(v)}{l-K}\\
					& g'_{l+1}(v) > g'_{l+1}(u_{l+1}) \text{ since } v>u_{l+1}\\
					& \frac{\sum_{i=1}^{l}g_i'(u)}{l-K}>\frac{\sum_{i=1}^{l}g_i'(v)}{l-K}  \text{ since } u>v\\
					& \text{From first line we have } \frac{\sum_{i=1}^{l}g_i'(u)}{l-K} = g_{l+1}'(u_{l+1})
				\end{align*}
			\end{proof}
			\begin{lemma}
				Violation Q occurs to the right of Q.
			\end{lemma}
			
			\begin{proof}
				We have $l$ as violation Q and we will prove that $l+1$ will also be violation Q. We have $u>v>u_{l+1}$. This proof can be done in the same way as done above because above we are only using the fact that
				$u>v>u_{l+1}$ and structure slope is equal to $g'_{l+1}(u_{l+1})$.
			\end{proof}
			This proves the block structure of errors which in-turn proves Lemma \ref{lemma: PQerrors}.
			
				%
			\end{appendix}
		\end{document}